\documentclass[twoside]{article}%
\usepackage{amssymb}
\usepackage{amsfonts}
\usepackage{amsmath}
\usepackage{graphicx}%
\providecommand{\U}[1]{\protect\rule{.1in}{.1in}}
\newtheorem{theorem}{Theorem}

\newtheorem{remark}[theorem]{Remark}

\newenvironment{proof}[1][Proof]{\noindent\textbf{#1.} }{\ \rule{0.5em}{0.5em}}
\begin{document}

\title{\textbf{Blowup for the }$C^{1}$\textbf{ Solutions of the Euler-Poisson
Equations of Gaseous Stars in }$R^{N}$}
\author{M\textsc{anwai Yuen\thanks{E-mail address: nevetsyuen@hotmail.com }}\\\textit{Department of Applied Mathematics,}\\\textit{The Hong Kong Polytechnic University,}\\\textit{Hung Hom, Kowloon, Hong Kong}}
\date{Revised 21-Jan-2010}
\maketitle

\begin{abstract}
The Newtonian Euler-Poisson equations with attractive forces are the classical
models for the evolution of gaseous stars and galaxies in astrophysics. In
this paper, we use the integration method to study the blowup problem of the
$N$-dimensional system with adiabatic exponent $\gamma>1$, in radial symmetry.
We could show that the $C^{1}$ non-trivial classical solutions $(\rho,V)$,
with compact support in $[0,R]$, where $R>0$ is a positive constant with
$\rho(t,r)=0$ and $V(t,r)=0$ for $r\geq R$, under the initial condition
\begin{equation}
H_{0}=\int_{0}^{R}r^{n}V_{0}dr>\sqrt{\frac{2R^{2n-N+4}M}{n(n+1)(n-N+2)}}%
\end{equation}
with an arbitrary constant $n>\max(N-2,0),$\newline blow up before a finite
time $T$ for pressureless fluids or $\gamma>1.$ Our results could fill some
gaps about the blowup phenomena to the classical $C^{1}$ solutions of that
attractive system with pressure under the first boundary condition.\newline In
addition, the corresponding result for the repulsive systems is also provided.
Here our result fully covers the previous case for $n=1$ in "M.W. Yuen,
\textit{Blowup for the Euler and Euler-Poisson Equations with Repulsive
Forces}, Nonlinear Analysis Series A: Theory, Methods \& Applications 74
(2011), 1465--1470".



\end{abstract}

\section{Introduction}

The compressible isentropic Euler $(\delta=0)$ or Euler-Poisson $(\delta
=\pm1)$ equations can be written in the following form:
\begin{equation}
\left\{
\begin{array}
[c]{rl}%
{\normalsize \rho}_{t}{\normalsize +\nabla\cdot(\rho u)} & {\normalsize =}%
{\normalsize 0}\\
\rho\lbrack u_{t}+(u\cdot\nabla)u]{\normalsize +\nabla}P & {\normalsize =}%
{\normalsize \rho\nabla\Phi}\\
{\normalsize \Delta\Phi(t,x)} & {\normalsize =\delta\alpha(N)}%
{\normalsize \rho}%
\end{array}
\right.  \label{Euler-Poisson}%
\end{equation}
where $\alpha(N)$ is a constant related to the unit ball in $R^{N}$: :
$\alpha(1)=1$; $\alpha(2)=2\pi$ and for $N\geq3,$%
\begin{equation}
\alpha(N)=N(N-2)Vol(N)=N(N-2)\frac{\pi^{N/2}}{\Gamma(N/2+1)},
\end{equation}
where $Vol(N)$ is the volume of the unit ball in $R^{N}$ and $\Gamma$ is a
Gamma function.. As usual, $\rho=\rho(t,x)\geq0$ and $u=u(t,x)\in
\mathbf{R}^{N}$ are the density and the velocity respectively. $P=P(\rho)$\ is
the pressure function. The $\gamma$-law for the pressure term $P(\rho)$ could
be applied:%
\begin{equation}
{\normalsize P}\left(  \rho\right)  {\normalsize =K\rho}^{\gamma}\label{gamma}%
\end{equation}
which the constant $\gamma\geq1$. If\ ${\normalsize K>0}$, we call the system
with pressure; if ${\normalsize K=0}$, we call it pressureless.\newline When
$\delta=-1$, the system is self-attractive. The equations (\ref{Euler-Poisson}%
) are the Newtonian descriptions of gaseous stars or a galaxy in astrophysics
\cite{BT} and \cite{C}. When $\delta=1$, the system is the compressible
Euler-Poisson equations with repulsive forces. It can be used as a
semiconductor model \cite{Cse}. For the compressible Euler equation with
$\delta=0$, it is a standard model in fluid mechanics \cite{Lions}. And the
Poisson equation (\ref{Euler-Poisson})$_{3}$ could be solved by%
\begin{equation}
{\normalsize \Phi(t,x)=\delta}\int_{R^{N}}G(x-y)\rho(t,y){\normalsize dy}%
\end{equation}
where $G$ is Green's function:
\begin{equation}
G(x)\doteq\left\{
\begin{array}
[c]{ll}%
|x|, & N=1\\
\log|x|, & N=2\\
\frac{-1}{|x|^{N-2}}, & N\geq3.
\end{array}
\right.
\end{equation}

Here, the solutions in radial symmetry could be:
\begin{equation}
\rho=\rho(t,r)\text{ and }u=\frac{x}{r}V(t,r)=:\frac{x}{r}V
\end{equation}
with the radius $r=\left(  \sum_{i=1}^{N}x_{i}^{2}\right)  ^{1/2}$.

The Poisson equation (\ref{Euler-Poisson})$_{3}$ becomes%
\begin{equation}
{\normalsize r^{N-1}\Phi}_{rr}\left(  {\normalsize t,x}\right)  +\left(
N-1\right)  r^{N-2}\Phi_{r}{\normalsize =}\alpha\left(  N\right)
\delta{\normalsize \rho r^{N-1}}%
\end{equation}%
\begin{equation}
\Phi_{r}=\frac{\alpha\left(  N\right)  \delta}{r^{N-1}}\int_{0}^{r}%
\rho(t,s)s^{N-1}ds.
\end{equation}
By standard computation, the systems in radial symmetry can be rewritten in
the following form:%
\begin{equation}
\left\{
\begin{array}
[c]{c}%
\rho_{t}+V\rho_{r}+\rho V_{r}+\dfrac{N-1}{r}\rho V=0\\
\rho\left(  V_{t}+VV_{r}\right)  +P_{r}(\rho)=\rho\Phi_{r}\left(  \rho\right)
.
\end{array}
\right.  \label{eq12345}%
\end{equation}

In literature for constructing analytical solutions for these systems,
interested readers could refer to \cite{GW}, \cite{M1}, \cite{DXY}, \cite{Y1}
and \cite{Y4}. The local existence for the systems can be found in
\cite{Lions}, \cite{M2}, \cite{B} and \cite{G}. The analysis of stabilities
for the systems may be referred to \cite{E}, \cite{MUK}, \cite{MP}, \cite{P},
\cite{DLY}, \cite{DXY}, \cite{Y2}, \cite{CT}, \cite{CB} and \cite{YuenNA}.

In literature for showing blowup results for the solutions of these systems,
Makino, Ukai and Kawashima firstly defined the tame solutions \cite{MUK} for
outside the compact of the solutions%
\begin{equation}
V_{t}+VV_{r}=0.
\end{equation}
After this, Makino and Perthame continued the blowup studies of the "tame"
solutions for the Euler system with gravitational forces \cite{MP}. Then
Perthame proved the blowup results for $3$-dimensional pressureless system
with repulsive forces \cite{P} $(\delta=1)$. In fact, all these results rely
on the solutions with radial symmetry:
\begin{equation}
V_{t}+VV_{r}{\normalsize =}\frac{\alpha(N)\delta}{r^{N-1}}\int_{0}^{r}%
\rho(t,s)s^{N-1}ds.
\end{equation}
And the Emden ordinary differential equations were deduced on the boundary
point of the solutions with compact support:%
\begin{equation}
\frac{D^{2}R}{Dt^{2}}=\frac{\delta M}{R^{N-1}},\text{ }R(0,R_{0})=R_{0}%
\geq0,\text{ }\dot{R}(0,R_{0})=0
\end{equation}
where $\frac{dR}{dt}:=V$ and $M$ is the mass of the solutions, along the
characteristic curve. They showed the blowup results for the $C^{1}$ solutions
of the system (\ref{eq12345}).

In 2008 and 2009, Chae, Tadmor and Cheng in \cite{CT} and \cite{CB} showed the
finite time blowup, for the pressureless Euler-Poisson equations with
attractive forces $(\delta=-1)$, under the initial condition,%
\begin{equation}
S:=\{\left.  a\in R^{N}\right\vert \text{ }\rho_{0}(a)>0,\text{ }\Omega
_{0}(a)=0,\text{ }\nabla\cdot u(0,x(0)<0\}\neq\phi\label{chea}%
\end{equation}
where $\Omega$ is the rescaled vorticity matrix $(\Omega_{_{0}ij})=\frac{1}%
{2}(\partial_{i}u_{0}^{j}-\partial_{j}u_{0}^{i})$ with the notation
$u=(u^{1},u^{2},....,u^{N})$ in their paper and some point $x_{0}$.

They used the analysis of spectral dynamics to show the Riccati differential
inequality,%
\begin{equation}
\frac{D\operatorname{div}u}{Dt}\leq-\frac{1}{N}(\operatorname{div}u)^{2}.
\label{ineq1}%
\end{equation}
The $C^{2}$ solution for the inequality (\ref{ineq1}) blows up on or before
$T=-N/(\nabla\cdot u(0,x_{0}(0))$. But, their method cannot be applied to the
system with the system with pressure or the classical $C^{1}$ solutions for
the system without pressure.

We are in particular interested in the instabilities of the Euler-Poisson
equations in $R^{3}$. For the instabilities of the global classical solutions
for $\gamma>4/3$ is shown by the second inertia functional
\begin{equation}
H(t)=\int_{\Omega}\rho(t,x)\left\vert x\right\vert ^{2}dx
\end{equation}
with the constant total energy
\begin{equation}
E_{0}=E(t)=\int_{\Omega}\left(  \frac{1}{2}\rho\left\vert u\right\vert
^{2}+\frac{K}{\gamma-1}\rho^{\gamma}+\frac{1}{2}\rho\Phi\right)  dx
\end{equation}
\cite{MP} and \cite{DLY}. For the critical case $\gamma=4/3$, any small
perturbation which makes the energy positive would cause that the solutions go
to positive infinity, which implies an instability of such a state. And the
instability for the stationary solution for $\gamma=6/5$ is obtained \cite{J}.
But the nonlinear stability result for the other cases is unknown. It was
shown that any stationary solution is stable for $\gamma\in(4/3,2)$ and
unstable for $\gamma\in(1,4/3)$ in linearized case \cite{BT} and \cite{L}.

On the other hand, Yuen \cite{YuenNA} in 2010 used the integration method to
show that with the initial velocity
\begin{equation}
H_{0}=\int_{0}^{R}rV_{0}dr>0\label{condition1}%
\end{equation}
the solutions with compact support to the Euler $(\delta=0)$ or Euler-Poisson
equations with repulsive forces $(\delta=1)$ blow up before the finite time
$T$.

In this article, we could further apply the integration method in
\cite{YuenNA} for the system with attractive forces $(\delta=-1)$ to fill some
gaps about the blowup phenomena for the classical $C^{1}$ solutions of the
system with the pressure $\gamma>1$ or pressureless fluids:

\begin{theorem}
\label{thm:1 copy(1)}Consider the Euler-Poisson equations with attractive
forces $(\delta=-1)$ (\ref{eq12345}) in $R^{N}$. The non-trivial $C^{1}$
classical solutions $\left(  \rho,V\right)  $, in radial symmetry, with
compact support in $\left[  0,R\right]  $, where $R>0$ is a positive constant
($\rho(t,r)=0$ and $V(t,r)=0$ for $r\geq R$) and the initial velocity:
\begin{equation}
H_{0}=\int_{0}^{R}r^{n}V_{0}dr>\sqrt{\frac{2R^{2n-N+4}M}{n(n+1)(n-N+2)}}%
\end{equation}
with an arbitrary constant $n>\max(N-2,0)$ and $M$ is the total mass of the
fluids,\newline blow up before a finite time $T$ for pressureless fluids
$(K=0)$ or $\gamma>1$.
\end{theorem}

Here, the boundary condition for the fluids:%
\begin{equation}
\rho(t,r)=0\text{ and }V(t,r)=0\text{ for }r\geq R
\end{equation}
is called non-slip boundary condition \cite{Day} and \cite{CTC}.

\section{Integration Method}

We just modify the integration method which was initially designed for the
Euler-Poisson system with repulsive force $(\delta=1)$ in \cite{YuenNA} to
obtain the corresponding blow results for the ones with attractive forces
$(\delta=-1)$.

\begin{proof}
We only use the density function $\rho(t,x(t;x))$ to preserve its non-negative
nature as we take integration for the mass equation (\ref{Euler-Poisson}%
)$_{1}$:%
\begin{equation}
\frac{D\rho}{Dt}+\rho\nabla\cdot u=0 \label{eqq2}%
\end{equation}
with the material derivative,%
\begin{equation}
\frac{D}{Dt}=\frac{\partial}{\partial t}+\left(  u\cdot\nabla\right)
\label{eqq1}%
\end{equation}
to have:%
\begin{equation}
\rho(t,x)=\rho_{0}(x_{0}(0,x_{0}))\exp\left(  -\int_{0}^{t}\nabla\cdot
u(t,x(t;0,x_{0}))dt\right)  \geq0
\end{equation}
for $\rho_{0}(x_{0}(0,x_{0}))\geq0,$ along the characteristic curve.

Then we could control the momentum equation (\ref{eq12345})$_{2}$ for the
non-trivial solutions in radial symmetry, $\rho_{0}\neq0$, to obtain:%
\begin{equation}
V_{t}+VV_{r}+K\gamma\rho^{\gamma-2}\rho_{r}=\Phi_{r}%
\end{equation}%
\begin{equation}
V_{t}+\frac{\partial}{\partial r}(\frac{1}{2}V^{2})+K\gamma\rho^{\gamma-2}%
\rho_{r}=\Phi_{r}\label{eq478}%
\end{equation}%
\begin{equation}
r^{n}V_{t}+r^{n}\frac{\partial}{\partial r}(\frac{1}{2}V^{2})+K\gamma
r^{n}\rho^{\gamma-2}\rho_{r}=r^{n}\Phi_{r}\label{eq789}%
\end{equation}
with multiplying the function $r^{n}$ with $n>0$, on the both sides.\newline
We could take the integration with respect to $r,$ to equation (\ref{eq789}),
for $\gamma>1$ or $K\geq0$:%
\begin{equation}
\int_{0}^{R}r^{n}V_{t}dr+\int_{0}^{R}r^{n}\frac{d}{dr}(\frac{1}{2}V^{2}%
)+\int_{0}^{R}K\gamma r^{n}\rho^{\gamma-2}\rho_{r}dr=-\int_{0}^{R}r^{n}%
\Phi_{r}dr
\end{equation}%
\begin{equation}
\int_{0}^{R}r^{n}V_{t}dr+\int_{0}^{R}r^{n}\frac{d}{dr}(\frac{1}{2}V^{2}%
)+\int_{0}^{R}\frac{K\gamma r^{n}}{\gamma-1}d\rho^{\gamma-1}=-\int_{0}%
^{R}\left[  \frac{\alpha(N)r^{n}}{r^{N-1}}\int_{0}^{r}\rho(t,s)s^{N-1}%
ds\right]  dr\label{789}%
\end{equation}
with the estimation for the right hand side of (\ref{789}):
\begin{equation}
\int_{0}^{R}\left[  \frac{\alpha(N)r^{n}}{r^{N-1}}\int_{0}^{r}\rho
(t,s)s^{N-1}ds\right]  dr\leq\int_{0}^{R}\left[  \frac{\alpha(N)r^{n}}%
{r^{N-1}}\int_{0}^{R}\rho(t,s)s^{N-1}ds\right]  dr=\int_{0}^{R}r^{n-N+1}%
Mdr=\frac{R^{n-N+2}M}{n-N+2}%
\end{equation}
where $M$ is the total mass of the fluids and the constant $n>\max(N-2,0)$,%
\begin{equation}
-\int_{0}^{R}\left[  \frac{\alpha(N)r^{n}}{r^{N-1}}\int_{0}^{r}\rho
(t,s)s^{N-1}ds\right]  dr\geq-\frac{R^{n-N+2}M}{n-N+2}%
\end{equation}
to have%
\begin{equation}
\int_{0}^{R}r^{n}V_{t}dr+\int_{0}^{R}r^{n}\frac{d}{dr}(\frac{1}{2}V^{2}%
)+\int_{0}^{R}\frac{K\gamma r^{n}}{\gamma-1}d\rho^{\gamma-1}\geq
-\frac{R^{n-N+2}M}{n-N+2}.\label{eq567}%
\end{equation}
We notice that this is the critical step in this paper to obtain the
corresponding blowup results for the Euler system with attractive forces.

Then, the below equation could be showed by integration by part:%
\begin{equation}%
\begin{array}
[c]{c}%
\int_{0}^{R}r^{n}V_{t}dr-\frac{1}{2}\int_{0}^{R}nr^{n-1}V^{2}dr+\frac{1}%
{2}\left[  R^{n}V^{2}(t,R)-0^{n}\cdot V^{2}(t,0)\right] \\
-\int_{0}^{R}\frac{K\gamma nr^{n-1}}{\gamma-1}\rho^{\gamma-1}dr+\frac{K\gamma
}{\gamma-1}\left[  R^{n}\rho^{\gamma-1}(t,R)-0^{n}\cdot\rho^{\gamma
-1}(t,0)\right]  \geq-\frac{R^{n-N+2}M}{n-N+2}.
\end{array}
\end{equation}
The above inequality with the boundary condition with a uniform compact
support ($V(t,R)=0$ and $\rho(t,R)=0$), becomes%
\begin{equation}
\int_{0}^{R}r^{n}V_{t}dr-\frac{1}{2}\int_{0}^{R}nr^{n-1}V^{2}dr-\int_{0}%
^{R}\frac{K\gamma nr^{n-1}}{\gamma-1}\rho^{\gamma-1}dr\geq-\frac{R^{n-N+2}%
M}{n-N+2}%
\end{equation}%
\begin{equation}
\frac{d}{dt}\int_{0}^{R}r^{n}Vdr-\frac{1}{2}\int_{0}^{R}nr^{n-1}V^{2}%
dr-\int_{0}^{R}\frac{K\gamma nr^{n-1}}{\gamma-1}\rho^{\gamma-1}dr\geq
-\frac{R^{n-N+2}M}{n-N+2}%
\end{equation}%
\begin{equation}
\frac{d}{dt}\frac{1}{n+1}\int_{0}^{R}Vdr^{n+1}-\frac{1}{2}\int_{0}^{R}\frac
{n}{(n+1)r}V^{2}dr^{n+1}+\frac{R^{n-N+2}M}{n-N+2}\geq\int_{0}^{R}\frac{K\gamma
nr^{n-1}}{\gamma-1}\rho^{\gamma-1}dr\geq0
\end{equation}
for $n>0$ and $\gamma>1$ or $K=0.$\newline For non-trivial initial density
functions $\rho_{0}\geq0$, we have:%
\begin{equation}
\frac{d}{dt}\frac{1}{n+1}\int_{0}^{R}Vdr^{n+1}-\frac{1}{2}\int_{0}^{R}\frac
{n}{(n+1)r}V^{2}dr^{n+1}+\frac{R^{n-N+2}M}{n-N+2}\geq0
\end{equation}%
\begin{equation}
\frac{d}{dt}\frac{1}{n+1}\int_{0}^{R}Vdr^{n+1}+\frac{R^{n-N+2}M}{n-N+2}%
\geq\int_{0}^{R}\frac{n}{2(n+1)r}V^{2}dr^{n+1}\geq\frac{n}{2(n+1)R}\int
_{0}^{R}V^{2}dr^{n+1} \label{eq11111}%
\end{equation}%
\begin{equation}
\frac{d}{dt}\int_{0}^{R}Vdr^{n+1}+\frac{R^{n-N+2}M}{n-N+2}\geq\frac{n}{2R}%
\int_{0}^{R}V^{2}dr^{n+1}. \label{eq2222}%
\end{equation}
We could denote
\begin{equation}
H:=H(t)=\int_{0}^{R}r^{n}Vdr=\frac{1}{n+1}\int_{0}^{R}Vdr^{n+1}%
\end{equation}
and apply the Cauchy-Schwarz inequality:%
\begin{equation}
\left\vert \int_{0}^{R}V\cdot1dr^{n+1}\right\vert \leq\left(  \int_{0}%
^{R}V^{2}dr^{n+1}\right)  ^{1/2}\left(  \int_{0}^{R}1dr^{n+1}\right)  ^{1/2}%
\end{equation}%
\begin{equation}
\left\vert \int_{0}^{R}V\cdot1dr^{n+1}\right\vert \leq\left(  \int_{0}%
^{R}V^{2}dr^{n+1}\right)  ^{1/2}\left(  R^{n+1}\right)  ^{1/2}%
\end{equation}%
\begin{equation}
\frac{\left\vert \int_{0}^{R}Vdr^{n+1}\right\vert }{R^{\frac{n+1}{2}}}%
\leq\left(  \int_{0}^{R}V^{2}dr^{n+1}\right)  ^{1/2}%
\end{equation}%
\begin{equation}
\frac{(n+1)^{2}H^{2}}{R^{n+1}}\leq\int_{0}^{R}V^{2}dr^{n+1}%
\end{equation}%
\begin{equation}
\frac{n(n+1)^{2}H^{2}}{2R^{n+2}}\leq\frac{n}{2R}\int_{0}^{R}V^{2}dr^{n+1}
\label{ineq123}%
\end{equation}
for letting equation (\ref{eq2222}) to be%
\begin{equation}
\frac{d}{dt}(n+1)H+\frac{R^{n-N+2}M}{n-N+2}\geq\frac{n}{2R}\int_{0}^{R}%
V^{2}dr^{n+1}\geq\frac{n(n+1)^{2}H^{2}}{2R^{n+2}}%
\end{equation}
with inequality (\ref{ineq123}),%
\begin{equation}
\frac{d}{dt}H\geq\frac{n(n+1)H^{2}}{2R^{n+2}}-\frac{R^{n-N+2}M}{n-N+2}.
\label{Ricc}%
\end{equation}
If we require the initial condition%
\begin{equation}
H_{0}^{{}}>\sqrt{\frac{2R^{2n-N+4}M}{n(n+1)(n-N+2)}},
\end{equation}
it is well-known for that the solutions for the Riccati differential
inequality (\ref{Ricc}) blow up before a finite time $T$:%
\begin{equation}
\underset{t->T^{-}}{\lim}H(t)=+\infty.
\end{equation}
Therefore, we could show that the $C^{1}$ solutions blow up before a finite
time $T.$

This completes the proof.
\end{proof}

Additionally, it is clear to see that we could further apply the integration
method to extend Yuen's result \cite{Day} as the following theorem:

\begin{theorem}
\label{thm:1}Consider the Euler $(\delta=0)$ or Euler-Poisson equations with
repulsive forces $(\delta=1)$ (\ref{Euler-Poisson}) in $R^{N}$. The
non-trivial classical solutions $\left(  \rho,V\right)  $, in radial symmetry,
with compact support in $\left[  0,R\right]  $, where $R>0$ is a positive
constant ($\rho(t,r)=0$ and $V(t,r)=0$ for $r\geq R$) and the initial
velocity:
\begin{equation}
H_{0}=\int_{0}^{R}r^{n}V_{0}dr>0
\end{equation}
with an arbitrary constant $n>0,$\newline blow up on or before the finite time
$T=2R^{n+2}/(n(n+1)H_{0}),$ for pressureless fluids $(K=0)$ or $\gamma>1$.
\end{theorem}

\begin{proof}
We just reverse the estimation for the equation (\ref{789}) with repulsive
forces $(\delta=1)$ to have
\begin{equation}
\int_{0}^{R}r^{n}V_{t}dr+\int_{0}^{R}r^{n}\frac{d}{dr}(\frac{1}{2}V^{2}%
)+\int_{0}^{R}\frac{K\gamma r^{n}}{\gamma-1}d\rho^{\gamma-1}=\int_{0}%
^{R}\left[  \frac{\alpha(N)r^{n}}{r^{N-1}}\int_{0}^{r}\rho(t,s)s^{N-1}%
ds\right]  dr
\end{equation}%
\begin{equation}
\int_{0}^{R}r^{n}V_{t}dr+\int_{0}^{R}r^{n}\frac{d}{dr}(\frac{1}{2}V^{2}%
)+\int_{0}^{R}\frac{K\gamma r^{n}}{\gamma-1}d\rho^{\gamma-1}\geq0
\end{equation}
for $\delta\geq0.$\newline For non-trivial initial density functions $\rho
_{0}\geq0$, we may obtain the corresponding result:%
\begin{equation}
\frac{d}{dt}H\geq\frac{n(n+1)H^{2}}{2R^{n+2}}%
\end{equation}%
\begin{equation}
H\geq\frac{-2R^{n+2}H_{0}}{n(n+1)H_{0}t-2R^{n+2}}.
\end{equation}
Then, we could require the initial condition
\begin{equation}
H_{0}=\int_{0}^{R}r^{n}V_{0}dr>0
\end{equation}
for showing that the solutions blow up on or before the finite time
$T=2R^{n+2}/(n(n+1)H_{0}).$

This completes the proof.
\end{proof}

We notice that Theorem 2 in this paper fully cover the previous case for $n=1$
in \cite{YuenNA}. Further researches are needed to have the corresponding
results for the non-radial symmetric cases.

\begin{remark}
If the global $C^{1}$ solutions with compact support whose radii expand
unboundedly when time goes to infinity, the results in this paper could not
offer the information about this case.
\end{remark}

\begin{remark}
The blowup results in this paper imply that the classical Euler-Poisson
equations even for the ones with attractive forces could not be used to be a
good modelling for the evolutions of gaseous stars with the relativistically
large functional $H_{0}$ which compares with the total mass of the fluid.
Alternatively, the relativistic Euler-Poisson equations may be adopted for
these cases to prevent these blowup phenomena.
\end{remark}

On the other hand, the author conjectures that there exists a variational
version for showing blowup phenomena in non-radially symmetrical cases with
the first boundary condition. But, another nice functional form would be
required to be designed to handle the corresponding problems.

\end{document}